\documentclass[11pt,a4paper]{amsart}
\usepackage[left=3.5cm,right=3.5cm,top=3.5cm,bottom=3.5cm]{geometry}
\usepackage{enumerate}
\usepackage{leftidx}
\usepackage{mathtools}
\usepackage{amsmath,amscd,amssymb,amsthm}
\usepackage[arrow,matrix]{xy}
\usepackage[OT2,T1]{fontenc}
\usepackage{booktabs}
\usepackage[position=bottom]{subfig}
\usepackage{lipsum}
\usepackage{graphicx}
\usepackage{setspace,color}
\usepackage{hyperref}
\numberwithin{equation}{section}
\numberwithin{table}{section}

\DeclareMathOperator{\Gal}{Gal}

\theoremstyle{definition}
\newtheorem{definition}{Definition}[section]
\newtheorem{example}[definition]{Example}
\newtheorem{remark}[definition]{Remark}
\newtheorem*{remark*}{Remark}

\theoremstyle{plain}
\newtheorem{theorem}[definition]{Theorem}

\newtheorem{lemma}[definition]{Lemma}
\newtheorem{proposition}[definition]{Proposition}

\newcommand{\F}{{ \mathbb F }}
\newcommand{\vF}{{ \mathbb F }}
\newcommand{\N}{{ \mathbb N }}
\newcommand{\vN}{{ \mathbb N }}
\newcommand{\vZ}{{ \mathbb N }}
\newcommand{\vQ}{{ \mathbb Q }}
\newcommand{\Q}{{ \mathbb Q }}
\newcommand{\Z}{{ \mathbb Z }}
\newcommand{\R}{{ \mathbb R }}
\newcommand{\p}{{ \mathfrak p }}

\renewcommand{\O}{{ \mathcal O }}

\author[A. Ferraguti]{Andrea Ferraguti}
\address{DICATAM, University of Brescia, via Branze 43, I-25123 Brescia}
\email{andrea.ferraguti@unibs.it}

\author[D. Goldfeld]{Dorian Goldfeld}
\address{Columbia University\\
2990 Broadway\\
New York, NY 10027 USA}
\email{goldfeld@columbia.edu}

\author[G. Micheli]{Giacomo Micheli}
\address{University of South Florida\\
4202 E Fowler Ave\\
33620 Tampa, US.
}
\email{gmicheli@usf.edu}
\title{Number Theoretical Locally Recoverable Codes}

\makeatletter
\@namedef{subjclassname@2020}{%
  \textup{2020} Mathematics Subject Classification}
\makeatother
\begin{document}

\thanks{Dorian Goldfeld is partially supported by Simons Collaboration Grant  567168.\\ Giacomo Micheli was supported in part by the National Science Foundation under Grant No 2127742.
}
\maketitle

{\flushright \footnotesize \emph{In honor of Joachim Rosenthal's 60th birthday.}\par}

\begin{abstract}
In this paper we give constructions for infinite sequences of finite non-linear locally recoverable codes $\mathcal C\subseteq \prod\limits^N_{i=1}\mathbb F_{q_i}$  over a product of finite fields arising from basis expansions in algebraic number fields.
The codes in our sequences have increasing length and size, constant rate, fixed locality, and minimum distance going to infinity.
\end{abstract}

\section{Introduction}

There has been a lot of interest recently in  Locally Recoverable Codes (LRC) \cite{barg2017locally,bartoli2020locally,dukes2022optimal,
freij2016locally,liu2018new,micheliIEEE,silberstein2013optimal,tamo2014family}, which are linear codes that allow local recovery of erasures. More specifically, they allow recovery of simultaneous erasures exactly as commonly used $k$-dimensional codes do (e.g.\ Reed-Solomon codes) by looking at the entries of a codeword that correspond to an information set (i.e.\ $k$ other components where no erasure happened), but they also allow recovery of a single erasure by looking at fewer nodes than $k$.

The applications in which LRC thrive are related to distributed storage and cloud storage systems because they easily allow the recovery of the data in a single failed server or hard-drive (which in this context we simply call \emph{node}), but they also allow recovery of more serious failures, such as simultaneous failures of multiple nodes in the system.

In this paper we construct Locally Recoverable Codes using Number Fields. The construction is inspired by Tamo and Barg's ideas in \cite{tamo2014family}, used in combination with the framework defined by Guruswami in \cite{guruswami}.  The technique to construct our codes is purely number theoretical, but the constructed codes are subsets of a product of finite fields (so they are practical to implement). We believe that  having LRCs over products of different finite fields  is an interesting feature as it provides more flexibility in the design of a system: for example, we might want to increase the length of these codes without having to enlarge the base field for all the components of every codeword, as one would be forced to do in the context of Reed-Solomon codes when one wants to extend the length of the code beyond the size of the finite field. In fact, with our construction, we can simply take additional reductions of elements of $\mathcal O_K$ modulo other prime ideals of $\mathcal O_K$.
Of course, our codes are not linear because of the nature of a product of finite fields and the fact that $(\mathcal O_K,+)$ does not have an $\vF_q$-linear structure for any $q$. Nevertheless, they allow efficient local recovery and have good minimum distance. 

In addition, we can construct a family of codes such that the distance grows linearly with length and dimension, and the asymptotic rate (see Definition \ref{good1}) can be made constant. From the methodology standpoint, our construction builds up new interactions between analytic number theory and coding theory. These new interactions stem from the idea that one can look at any integer $M$ as a constant function from the set of primes 
$\{\mathfrak p_i\}_{i\in \{1,\dots, n\}}$ 
of the ring of integers  lying over a totally split rational prime $p$, to 
$\vF_p$ simply by mapping each $\mathfrak p_i$ to the reduction of $M$ modulo $\mathfrak p_i$. Notice that since $\mathfrak p_i$ is totally split, $M \mod \mathfrak p_i=M\mod  p$, independently of $i$. This simple idea allows to build the locality sets, which correspond to totally split primes of the number field.

\section{Background on Coding Theory}

Let $n$ be a positive integer and $F_1,\dots F_n$ be finite fields ordered by increasing size.
We define a code $\mathcal C$ in $R_m=\prod\limits^m_{i=1} F_i$ as a subset of $R_m$.
The distance between two codewords $x,y\in \mathcal C$ is the number of indexes $i\in \{1,\dots n\}$ such that $x_i\neq y_i$.
The minimum distance of $\mathcal C$ is the minimal positive integer $d$ such that there exist two elements $x,y \in \mathcal C$ such that $d(x,y)=d$. If $X$ is a set, let us denote the powerset of $X$ as $2^X$.

For every $m\in \vN$, let $\mathcal C_m \in 2^{R_m}$ be a choice of a code of $R_m$. For any set $A$, let us denote by $\# A$ the cardinality of $A$.
We say that a sequence of codes $\{\mathcal C_m\}_{m\in \N}$ is \emph{almost good} if 
\begin{equation}\label{good1}
\liminf_{m\rightarrow \infty} \frac{\log (\# \mathcal C_m)}{\log (\# R_m)}=\gamma>0 
\end{equation}

and 

\begin{equation}\label{good2}
\liminf_{m\rightarrow \infty}d(\mathcal C_m)=+\infty. \end{equation}

Notice that the choice of the basis of the logarithm clearly does not affect the definition. Moreover, observe that $\frac{\log (\# \mathcal C_m)}{\log (\# R_m)}$ is the natural generalization of the concept of information rate in the non-linear setting, since the dimension of a non-linear code $\mathcal C$ is replaced by $\log(\#\mathcal C)$ and the dimension of the ambient space is $\log (\# R_m)$.

Also, notice the difference between our definition and the usual definition of a good family of codes, where the ratio between distance and length is required to converge to a constant. In our case, since the distance and locality are not weighted by how large the finite fields we are using are, the standard definition of good codes carries an inherent disadvantage that is essentially unavoidable. For this reason, we do keep track of the size of the code vs the size of the full space but for simplicity we avoid weighting the distance $d$ depending on the finite fields where the components belong, and only require $d\rightarrow \infty$ (even though in our case the growth is linear, which is what happens with optimal codes).

We say that a code $\mathcal C\subseteq R_m$ has locality $r$ if for any $i\in\{1,\dots,n\}$ it is possible to reconstruct the $i$-th component of a codeword $c\in \mathcal C$ by knowing at most $r$ other components of $c$. In other words, there is an algorithm (depending on $\mathcal C$)  that takes as input the location $i$ of an erasure together with $r$ other coordinates of $c$ and outputs the missing component of $c$.

\section{Background on Number Fields}

Let $K/\Q$ be a number field of degree $\delta$. Recall that if $\beta\in K$, the \emph{norm} of $\beta$, denoted by $N(\beta)$, is the determinant of the $\Q$-linear map $K\to K$ defined by $x\mapsto \beta x$. Let $\O_K$ be the ring of integers of $K$, and let $\alpha\in \O_K$ be an element that generates $K$, i.e.\ such that $\vQ(\alpha)=K$. Let
$$m_\alpha(x)=b_0+b_1x+\ldots+b_{\delta-1}x^{\delta-1}+x^\delta\in \Z[x]$$
be the minimal polynomial of $\alpha$ over $\Q$, and let $S\coloneqq \max\{|b_i|\colon i\in \{0,\ldots,\delta-1\}\}$, where $|\cdot|$ denotes the usual archimedean absolute value.
For any prime ideal $\p$ of $\mathcal O_K$, let $\vF_{\p}$ be the field $\mathcal O_K/\p$.

Recall that a prime $p\in \vZ$ is \emph{totally split} in $K/\vQ$ if $p\mathcal O_K$ factors as $\prod\limits^{[K:\vQ]}_{i=1} \p_i$, where the $\p_i$'s are pairwise distinct prime ideals of $\mathcal O_K$ and $[\vF_{\p_i}:\vF_p]=1$ for all $i\in \{1,\dots, \delta\}$.

\begin{lemma}\label{norm_bound}
With the notation above, let $y=\sum\limits^{\delta-1}_{i=0}z_i\alpha^i\in \O_K$, with $z_i\in \Z$ and $|z_i|< M$ for every $i$. Then $|N(y)|\leq \delta^{\delta/2}(1+S)^{(\delta-1)\delta/2}(M-1)^\delta$.
\end{lemma}
\begin{proof}
Clearly we can assume that $M>1$, as otherwise the claim is trivial. The set $\mathcal B\coloneqq \{1,\alpha,\ldots,\alpha^{\delta-1}\}$ is a $\Q$-basis of $K$ by assumption. Notice that if $w\coloneqq \sum\limits_{i=0}^{\delta-1}w_i\alpha^i\in \O_K$, with $w_i\in \Z$ for every $i$, then
\begin{align*}
w\cdot \alpha= & \sum_{i=0}^{\delta-1}w_i\alpha^{i+1} \\
			 = & \sum_{i=1}^{\delta-1}w_{i-1}\alpha^{i} + w_{\delta-1} \alpha^\delta.
\end{align*}

Using now that $m_\alpha(\alpha)=0$ we get that
\begin{equation}\label{coefficients}
w\cdot \alpha= -b_0w_{\delta-1}+\sum_{i=1}^{\delta-1}(w_{i-1}-b_iw_{\delta-1})\alpha^i.
\end{equation}

Let $A_y$ be the multiplication-by-$y$ matrix with respect to the basis $\mathcal B$, where the elements of $K$, expressed in the basis $\mathcal B$ are considered as column vectors. 
More precisely, $A_y$ is the matrix that makes the following diagram commutative
\[\begin{CD}
K @> y \cdot >> K\\
@VV \iota V @VV \iota V\\
\vQ^{\delta} @> A_y \cdot >> \vQ^{\delta}
\end{CD}
\]
where $\iota$ is the usual isomorphism of vector spaces that sends an element of $K$ into its expression in the basis $\mathcal B$. We claim that entries in the $j$-th column of $A_y$ are bounded, in absolute value, by $(M-1)(1+S)^{j-1}$. For $j=1$ this is obvious since entries in the first column are the coefficients of $y\cdot 1=y$ with respect to $\mathcal B$. Now suppose that the claim is true for the $j$-th column and let us prove it for the $(j+1)$-th. The $j$-th column is given by the result of the multiplication \[y\cdot \alpha^{j-1}=\sum_{i=0}^{\delta-1} c_i\alpha^i,\] with $|c_i|\leq (M-1)(1+S)^{j-1}$ for every $i\in \{0,\ldots,\delta-1\}$ by the inductive hypothesis. Now let us consider the $(j+1)$-th columns, given by the multiplication \[y\cdot \alpha^j=\sum_{i=0}^{\delta-1}d_i\alpha^i\] with $d_0,\ldots,d_{\delta-1}\in \Z$. Since $y\cdot \alpha^j=(y\cdot \alpha^{j-1})\cdot \alpha$, Equation \eqref{coefficients} and the inductive hypothesis show that $|d_0|\leq |b_0c_{\delta-1}|\leq S(1+S)^{j-1}$ and $|d_i|\leq |c_{i-1}|+|b_ic_{n-1}|\leq (M-1)(1+S)^j$. The claim follows since $S(1+S)^{j-1}\leq (M-1)(1+S)^j$ for every $j\geq 0$.

Now the bound on $|N(y)|=|\det A_y|$ follows from Hadamard's inequality, which states that the determinant of a complex matrix is bounded, in absolute value by the product of the euclidean norms of the column vectors $C_j$ of $A_y$. In fact,
\begin{align*}
\det A_y\leq\prod^{\delta}_{j=1} \| C_j\|\leq \prod^{\delta}_{j=1} \sqrt {\delta (M-1)^2(1+S)^{2(j-1)}}=\delta^{\delta/2} (M-1)^\delta (1+S)^{\delta(\delta-1)/2}.
\end{align*}

\end{proof}

\section{Construction of Number Theoretical Locally Recoverable Codes}

\subsection{Overview of the construction}

First, we construct an ambient code $\mathcal D$ (that is essentially a Chinese remainder code), for which we can prove nice distance properties. After that, we will extract a subcode of $\mathcal D$ that verifies the locality property we are seeking for. Finally we show how to construct almost good families of locally recoverable codes in the sense of Equation \eqref{good1} and Equation \eqref{good2}.

\subsection{Construction of the ambient code $\mathcal D$}
First, we need to construct Chinese remainder codes that are similar to the Reed-Solomon Codes (\`a-la Guruswami, see \cite{guruswami}).

Let $K$ be a number field of degree $\delta$ with ring of integers $\mathcal O_K$, and let $\alpha\in \O_K$ be such that $\vQ(\alpha)=K$. For $M\geq 1$ we define
\[\mathcal R[M]\coloneqq \left\{\sum^{\delta-1}_{i=0} z_i\alpha^i \mid 0\leq z_i< M, \quad \forall	i\in\{0,\dots \delta-1 \}\right\}.   \]
 Let $m_\alpha(x)=b_0+b_1x+\ldots+x^\delta\in \Z[x]$ be the minimal polynomial of $\alpha$ over $\Q$, and let $S\coloneqq \max\{|b_i|\colon i=0,\ldots,\delta-1\}$. Let 
 \begin{equation}\label{eq:Cdef}
 C_\alpha\coloneqq  \delta^{\delta/2}(1+S)^{(\delta-1)\delta/2}
 \end{equation}
 so that by Lemma \ref{norm_bound} we have that $|N(y)|\leq C_\alpha\cdot (M-1)^\delta$ for every $y\in \mathcal R[M]$.

Let $\mathfrak{p}_1,\dots \mathfrak{p}_n$ be distinct prime ideals of $\mathcal{O}_K$, ordered by increasing norm size, and for every $i$ let $\mathbb F_{\mathfrak{p}_i}\coloneqq\mathcal{O}_K/\mathfrak{p}_i$. Assume that $\prod\limits_{i=1}^nN(\p_i)>C_\alpha\cdot (M-1)^\delta$ (this is needed to achieve injectivity of the encoding map $\phi$ defined below).

The number theoretical Reed-Solomon code $\mathcal D=\mathcal D(K,M,\{\p_i\}_{i\in \{1,\dots n\}})$ is defined as the image $\phi(\mathcal R[M])$ of the map
\[\phi: \mathcal R[M] \longrightarrow \prod^{n}_{i=1}  \mathbb F_{\mathfrak{p}_i}\]
\[y \mapsto (y\bmod \p_1,\ldots,y\bmod \p_n).\]
See \cite{lenstra} for more on this.
In the rest of the paper we will refer to $\phi$ as the \emph{encoding map}.

\begin{theorem}\label{thm:reeddistance}
Let $\mathcal D$ be the code defined above, and let $\mathcal P\coloneqq \{\mathfrak p_1,\ldots, \mathfrak p_n\}$. Let $d(\mathcal D)$ be the minimal distance of $\mathcal D$ and let
$$m\coloneqq \min_{T\subseteq \mathcal P}\left\{\#T\colon\prod_{\p\in T}N(\p)>C_\alpha\cdot (M-1)^\delta\right\}.$$
Then the following hold.
\begin{enumerate}
\item The map $\phi$ is injective.
\item $d(\mathcal D)\geq n-m+1$.
\item If $\prod_{\p\in U}N(\p)<M^\delta$ for some $U\subseteq \mathcal P$ with $\# U=m-1$, then equality holds in (2).
\end{enumerate}
\end{theorem}
\begin{remark}
Since $\phi$ is thought as the encoding map, its injectivity is fundamental because we want that different messages are mapped to different codewords. This is achieved by adding just enough redundancy by considering at least $n$ distinct reductions, where $n$ is chosen such that $\prod\limits_{i=1}^nN(\p_i)>C_\alpha\cdot (M-1)^\delta$, as we will explain. 

Item (2) provides the code with a lower bound for the minimum distance: in fact the more redundancy is added (and therefore $n$ grows because we provide reductions at many prime ideals) the more the minimum distance grows. 

Item (3) ensures that, as far as the product of the norms is not too large then there are indeed two codewords at distance $m-1$.
\end{remark}
\begin{proof}
Let $T\subseteq \mathcal P$ be a subset of cardinality $m$ such that $\prod_{\p\in T}N(\p)>C_\alpha\cdot (M-1)^\delta$, and let $y_1,y_2\in\mathcal R[M]$ be such that $\phi(y_1)=\phi(y_2)$. In particular, we have that $y_1\equiv y_2\bmod \p$ for every $\p \in T$. It follows that $\prod_{\p\in T} N(\p)\mid N(y_1-y_2)$. By Lemma \ref{norm_bound} we have that $N(y_1-y_2)\leq C_\alpha\cdot (M-1)^\delta$, and hence by the definition of $m$ we must have that $y_1=y_2$. This proves (1) and (2) at the same time.

To prove (3), notice that if $\prod_{\p\in U}N(\p)<M^\delta$ then the map $\mathcal R[M]\to \prod_{\p \in U}\F_{\p}$ is not injective for cardinality reasons. It follows that there are $y_1\neq y_2\in\mathcal R[M]$ such that $y_1\equiv y_2\bmod \p$ for all $\p \in U$. On the other hand $\phi(y_1)\neq \phi(y_2)$ by (1), and hence $\phi(y_1)$ and $\phi(y_2)$ have distance $n-m+1$.
\end{proof}

%\begin{remark}
%The minimal distance can be strictly larger than $n-m+1$ if (3) is not satisfied. Consider for example the case where $K=\Q(i)$ with $\alpha=i$, $M=2$, $n=3$, $m=2$ and $\p_1=(3+2i)$. Then the map $\mathcal R[M]\to \F_{\p_1}$ is injective, showing that the minimal distance is $3>n-m+1$.
%\end{remark}

\subsection{Construction of the locally recoverable code $\mathcal C$ as a subset of $\mathcal D$}

Let $K$ be a number field of degree $r+1$ with ring of integers $\mathcal O_K$, and let $M\in \N$. Let $\alpha\in \mathcal O_K$ be such that $K=\Q(\alpha)$. Let 
\[\mathcal R[M]^{-}= \left\{\sum^{r-1}_{i=0} a_i \alpha^i \mid 0\leq a_i<M \quad \forall	i\in\{0,\dots r-1 \}\right\}\subsetneq \mathcal R[M].   \]
Notice that this differs from the set $\mathcal R[M]$ previously defined, as we are forcing the coefficient of $\alpha^r$ to be $0$ (this is a strictly smaller set of elements than $\mathcal R[M]$ since the minimal polyomial of $\alpha$ has degree $r+1$). Let now $s$ be a positive integer and define
$$\mathcal A[M]=\left\{\sum^s_{j=0} f_j M^j\colon f_j\in \mathcal R[M]^{-} \quad\forall   j\in \{0,\ldots,s\} \right\}.$$

\begin{lemma}
We have that
\[\# \mathcal A[M]=M^{r(s+1)}.\]
\end{lemma}
\begin{proof}
This follows from the fact that elements of $\mathcal R[M]$ are a complete set of representatives for the quotient $\O_K/(M)$. Hence if $\sum\limits_{j=0}^{r-1}f_jM^j=\sum\limits_{j=0}^{r-1} g_jM^j$ then $f_0\equiv g_0\bmod M$, but this implies that $f_0=g_0$. The claim follows then by an easy induction.
\end{proof}

Let $\ell$ be a positive integer, and let $p_1< p_2<\ldots < p_ \ell$ be rational primes that are totally split in $K:\vQ$. Suppose moreover that no $p_i$ divides the discriminant of the minimal polynomial of $\alpha$. For every $i\in\{1,\dots \ell\}$ let $\mathfrak{p}^{(p_i)}_1,\ldots, \mathfrak{p}^{(p_i)}_{r+1}$ be the prime ideals of $\mathcal O_K$ that lie above $p_i$. Notice that $\F_{\mathfrak{p}^{(p_i)}_{j}}=\F_{p_i}$ for all $i,j$.

The number theoretical locally recoverable code $\mathcal C=\mathcal C\left(r,s,K,M,\{p_i\}_{i\in \{1,\dots \ell\}}\right)$ is defined as the image $\phi(\mathcal A[M])$ of the map
\[\phi: \mathcal A[M] \longrightarrow \prod^{\ell}_{i=1} \prod^{r+1}_{j=1} \mathbb F_{\mathfrak{p}^{(p_i)}_{j}}= \prod^{\ell}_{i=1} \mathbb F_{p_i}^{r+1}\]
\[x\mapsto (x\bmod \p_j^{(p_i)})_{(i,j)\in U\times V},\]
where $U=\{1,\ldots,\ell\}$ and $V=\{1,\ldots,r+1\}$. For simplicity of notation, let us define $c_j^{(i)}=x\bmod \p_j^{(p_i)}$. Notice that \[\mathcal C\left(r,s,K,M,\{p_i\}_{i\in \{1,\dots \ell\}}\right)\subseteq \mathcal D\left(K,M^{s+1},\{\p^{(p_i)}_j\}_{i,j}\right).\]

\begin{lemma}\label{prop:codesize}
Assume that
$$\prod_{i\in \{1,\ldots,\ell\},j\in \{1,\ldots,r+1\}}N(\p_j^{(p_i)})>C_\alpha\cdot (M^{s+1}-1)^{r+1},$$
where $C_\alpha$ is the constant \eqref{eq:Cdef}. Then the code $\mathcal C =\phi(\mathcal A[M])$ has size $\#\mathcal A[M]=M^{r(s+1)}$, i.e. $\phi$ is injective.
\end{lemma}
\begin{proof}
To see this, simply notice that $\mathcal A[M]\subseteq \mathcal{R}[M^{s+1}]$ and then apply Theorem \ref{thm:reeddistance}.
\end{proof}

\begin{definition}
Whenever the hypothesis of Lemma \ref{prop:codesize} are verified, we say that the code $\mathcal C(r,s,K,M,\{p_i\}_{i\in \{1,\dots \ell\}})$ is a \emph{good split code} of length $n=(r+1)\ell$ and size $M^{r(s+1)}$ over the number field $K$.
\end{definition}

\begin{proposition}\label{prop:locality}
Let $\mathcal C=\mathcal C(r,s,K,M,\{p_i\}_{i\in \{1,\dots \ell\}})$ be a good split code. Then $\mathcal C$ has locality $r$. 
\end{proposition}
\begin{proof}
Suppose that the component $c_k^{(h)}$ of the codeword $c=\big(c_j^{(i)}\big)_{j=1,\ldots,r+1}^{i=1,\ldots,\ell}$ has to be retrieved. Such codeword $c$ arises from a message $m=\sum\limits^s_{j=0} f_j M^j\in \mathcal A[M]$, where $f_j=\sum\limits^{r-1}_{i=0} a_{i,j} \alpha^i\in\mathcal R[M]^-$. Now consider the components $c_1^{(h)},\ldots,c_{k-1}^{(h)},c_{k+1}^{(h)},\ldots,c_{r+1}^{(h)}$ of the codeword $c$. Each of them arises as the reduction of $m$ modulo $\p_j^{(p_h)}$, for some $j\in \{1,\ldots,k-1,k+1,\ldots,r+1\}$. The key point is now the following: since $p_h$ is totally split in $K$ and it does not divide the discriminant of the minimal polynomial $m_\alpha(x)$ of $\alpha$, by the Dedekind criterion we have that $m_\alpha(x)\bmod p_h=\prod\limits_{i=1}^{r+1}(x-\beta_i) \in \F_{p_h}[x]$ where $\beta_1,\ldots,\beta_{r+1}\in \F_{p_h}$ are pairwise distinct elements and $\beta_j$ is the image of $\alpha$ via the reduction map $\O_K\twoheadrightarrow \O_K/\p_j^{(p_h)}\cong \F_{p_h}$. Since $m=\sum\limits_{i=0}^{r-1} u_i\alpha^i$ for some integers $u_0,\ldots,u_ {r-1}$, the component $c_j^{(h)}$ of $c$ can be written as $\sum\limits_{i=0}^{r-1}u_i\beta_j^i$. This gives us a system of linear equations in $\F_{p_h}$, whose indeterminates are the reductions $\widetilde{u}_0,\ldots,\widetilde{u}_{r-1}$ of the $u_i$'s modulo $p_h$:
$$\begin{cases}\widetilde{u}_0+\widetilde{u}_1\beta_1+\ldots+\widetilde{u}_{r-1}\beta_1^{r-1}=c_1^{(h)} & \\
		\widetilde{u}_0+\widetilde{u}_1\beta_2+\ldots+\widetilde{u}_{r-1}\beta_2^{r-1}=c_2^{(h)} & \\
\ldots & \\
\widetilde{u}_0+\widetilde{u}_1\beta_{r+1}+\ldots+\widetilde{u}_{r-1}\beta_{r+1}^{r-1}=c_{r+1}^{(h)} & \\
\end{cases}.$$
Notice that the $k$-th row is deleted, since $c_k^{(h)}$ is missing. This is a system of $r$ equations in $r$ indeterminates, and its determinant is non-zero because the $\beta_i$'s are pairwise distinct and the matrix representing the system is a Vandermonde matrix. If $(v_0,\ldots,v_{r-1})\in \F_{p_h}^r$ is its unique solution, then $c_k^{(h)}=\sum\limits_{i=0}^{r-1}v_i\beta_k^i$.
\end{proof}

%We can identify $m$ with the map
%$$\{\mathfrak{p}^{(p_h)}_1,\ldots, \mathfrak{p}^{(p_h)}_{r+1}\}\to\prod^{r+1}_{j=1} \mathbb F_{p_h}$$
%that sends each $\mathfrak{p}^{(p_h)}_j$ to $m\bmod\mathfrak{p}^{(p_h)}_j$.
%Since all the $\mathfrak{p}^{(p_h)}_j$ lie over the same $p_h$,  $m$ evaluates as the map 
%\[m \mod p_h= \sum^s_{j=0} \sum^{r-1}_{i=0} a_{i,j} \alpha^i p^j\mod p_h\in \mathcal O_K/p\]. On the other hamd, $m\mod p_h$ can be written as $\sum^{r-1}_{i=0} u_i \alpha^i \mod p_h$. If now we can recover the $u_i$'s from $c_{p_h,k}$, for $k\in\{1,\dots i-1,i+1,r+1\}$ we are done, as we can then reduce $\sum^{r-1}_{i=0} u_i \alpha^i$ at $\mathfrak{p}^{(p_h)}_k$ to recover $c_{p_h,k}$. But this is easy, as one needs to solve the system given by the $r$ equations given by $\sum^{r-1}_{i=0} u_i \alpha^i \mod \mathfrak{p}^{(p_h)}_j=c_{p_h,j}$. These are linear equations over $\F_{p_h}$ as the reduction of $\alpha\mod \mathfrak{p}^{(p_h)}_j$  is an element of $\F_{p_h}$. These elements are all distinct as the minimal polynomial of $\alpha$, $m_\alpha$, is totally split at $p_h$, by the choice of the primes. The obtained equations are all linearly independent as they are the evaluations at distinct elements of $\F_{p_i}$ of the polynomial $\sum^{r-1}_{i=0} u_i T^i\in \F_{p_i}[T]$. This recovers the $u_i$'s, and therefore $c_{p_h,k}$ by the reduction modulo $\mathfrak{p}^{(p_h)}_k$.

\begin{proposition}\label{prop:distance}
Let $\mathcal C=\mathcal C(r,s,K,M,\{p_i\}_{i\in \{1,\dots \ell\}})$ be a good split code, and let $\mathcal P=\{\p_j^{(p_i)}\colon i\in \{i,\ldots, \ell\}, \,\,j\in \{1,\ldots,r+1\}\}$. Let
$$m\coloneqq \min_{T\subseteq \mathcal P}\left\{\# T\colon\prod_{\p \in T}N(\p)>C_\alpha\cdot (M^{s+1}-1)^{r+1}\right\}.$$
Then $\mathcal C$ has minimum distance $d\geq (r+1)\ell-m+1$.
\end{proposition}
\begin{proof}
The code $\mathcal C$ is a subcode of the number theoretical Reed-Solomon Code with parameter $M^{s+1}$, so we simply apply again Theorem \ref{thm:reeddistance} with $n=(r+1)\ell$ (as our set of primes consists of $r+1$ primes on $\mathcal O_K$ lying above each of the $\ell$ rational primes).
\end{proof}

The following theorem summarizes what we proved until now

\begin{theorem}
Let $K$ be a number field of degree $r+1$, let $s$ be a positive integer, let $M\in \N$, and let $\{p_i\}_{i\in \{1,\dots \ell\}}$ be a set of rational primes that are totally split in $K/\Q$. Let $\mathcal C=\mathcal C(r,s,K,M,\{p_i\}_{i\in \{1,\dots \ell\}})$ be a good split code over $K$.
Then $\mathcal C$ has length $\ell(r+1)$, size $M^{r(s+1)}$, minimum distance at least $(r+1)\ell-m+1$, and locality $r$.
\end{theorem}
\begin{proof}
Simply combine Lemma \ref{prop:codesize} and Propositions \ref{prop:locality}, and \ref{prop:distance}.
\end{proof}

\begin{example}
Let us illustrate our construction with a working example. Let $K=\Q(\alpha)$, where $\alpha\coloneqq \zeta_{16}+\zeta_{16}^{-1}$ and $\zeta_{16}$ is a primitive $16$-th root of $1$. The field $K$ is the largest totally real subfield of $\Q(\zeta_{16})$, and the extension $K/\Q$ is cyclic of degree $4$ and has discriminant $2^{11}$. The minimal polynomial of $\alpha$ is $x^4-4x^2+2$. The constant \eqref{eq:Cdef} is given by $C_{\alpha}= 4^2(1+4)^{6}=16\cdot 5^6=250000$. Primes that split completely in $K$ are exactly those congruent to $\pm1$ modulo $16$. Let $p_1=17$, $p_2=31$ and $p_3=47$, so that $\ell=3$. Let $\p_1^{(p_i)},\ldots,\p_4^{(p_i)}$ be the primes of $K$ lying above $p_i$, for every $i$. Let $M=2$ and $s=3$. One computes that $17^4\cdot 31^4>C_\alpha(M^4-1)^4$, while $17^4\cdot 31^3<C_\alpha(M^4-1)^4$, so that $\mathcal C=\mathcal C(3,3,K,2,\{17,31,47\})$ is a good split code of size $2^{12}$ and minimum distance at least $12-8+1=5$ according to Theorem \ref{thm:reeddistance}. We have that $\mathcal R[2]^-=\{a_0+a_1\alpha+a_2\alpha^2\colon a_i\in \{0,1\}\}$ while $\mathcal A[2]=\{f_0+f_1\cdot 2+f_2\cdot 2^2+f_3\cdot 2^3\colon f_i\in \mathcal R[2]^-\}$. Notice that
the encoding map is fully defined by giving the image of $\alpha$, because the reduction maps from $\mathcal O_K$ to $\mathcal O_K/I$ are homomorphisms for any ideal $I\subseteq \mathcal O_K$.
To obtain the image of $\alpha$, notice that 
 $$x^4-4x^2+2\equiv\begin{cases}(x+5)(x+8)(x+9)(x+12) & \mod 17\\ (x+5)(x+14)(x+17)(x+26) & \mod 31\\ (x+3)(x+18)(x+29)(x+44) & \mod 47\end{cases},$$
and therefore  the encoding $\phi$ is simply defined by
$$\phi\colon \mathcal A[M] \to \F_{17}^4\times \F_{31}^4\times \F_{47}^4$$
 $$\alpha\mapsto (12,9,8,5;26,17,14,5;44,29,18,3).$$
 $$ f(\alpha) \mapsto (f(12),f(9),f(8),f(5);f(26),f(17),f(14),f(5);f(44),f(29),f(18),f(3)).$$

Using MAGMA\cite{magma}, one can compute that the actual minimum distance of $\mathcal C$ is $6$, strictly better than the bound that comes from Theorem \ref{thm:reeddistance}.\flushright 

\end{example}

\subsection{Almost good families of good split codes}
We will now show how to construct an almost good family of good split codes, in the sense of Equation \eqref{good1} and Equation \eqref{good2}. To do so, we first need the following analytic number theoretical lemma.
\begin{lemma}\label{lemma:asymptotictotsplit}
Let $K$ be a Galois extension of $\vQ$. For every $\ell\ge 1$, let $p_1,\dots, p_\ell$ be the first $\ell$ totally split primes of $K/\vQ$. Then 
\[\log\left(\prod^{\ell}_{i=1} p_i\right)\sim \ell \log \ell\]
as $\ell$ tends to infinity.
\end{lemma}
\begin{proof}
First, let us recall that, if $X$ is a positive integer, the totally split primes of $K/\vQ$ verify
\begin{equation}\label{eq:sumlogs}
\log{\left(\underset{ \text{totally split}}{\prod_{p\leq X}} p\right)}\sim \frac{X}{[K:\vQ]}.
\end{equation}
In addition, the Chebotarev Density Theorem ensures that the asymptotic formula for the $\ell$-th totally split prime is
\[p_\ell \sim \ell\log (\ell) [K:\vQ].\]
By setting $X=p_\ell$ in \eqref{eq:sumlogs} and relabeling the product we get
\[\log\left(\prod^{\ell}_{i=1} p_i\right)\sim \frac{\ell\log(\ell)[K:\vQ]}{[K:\vQ]}=\ell\log (\ell).\]
\end{proof}

\begin{remark}
There are secondary terms in the asymptotic formula for $p_\ell$, i.e.\ the $\ell$-th (rational) totally split prime, and these secondary terms are larger than the secondary terms in the asymptotic formula for the $\ell$-th rational prime.
\end{remark}

Let now $K$ be a number field of degree $r+1$ that is Galois over $\vQ$ (this restriction allows for easier computations in Theorem \ref{thm:goodfamilygalois}), let $\alpha\in \O_K$ be a generator, let $C_\alpha$ be the constant \eqref{eq:Cdef} and let $s\in \N$. Let $\{p_j\}_{j\in \N}$ be the strictly increasing sequence of primes that are totally split in $K/\Q$, and for every $\ell\ge 1$ let $P_\ell\coloneqq \prod\limits_{i=1}^\ell p_i$.

\begin{theorem}\label{thm:goodfamilygalois}
Let $0<c<1$, let $k\in \R^+$ be such that $k<1/\sqrt[r+1]{C_\alpha}$, and let $M_\ell\coloneqq \left\lfloor\sqrt[s+1]{k\cdot P_\ell/P_{\lfloor c\ell\rfloor}}\right\rfloor$. Let $\mathcal C_\ell\coloneqq \mathcal C\left(r,s,K,M_\ell,\{p_i\}_{i\in \{1,\dots \ell\}}\right)$. Then $\{\mathcal C_\ell\}_{\ell\ge 1}$ is an almost good family of good split codes.
\end{theorem}

\begin{proof}
To prove that $\mathcal C_\ell$ is a good split code, it is enough to show that 
\[N\left(\prod\limits_{p\in \mathcal A}\prod\limits^{r+1}_{j=1} \mathfrak p_j^{(p_i)} \right)=  P_\ell^{r+1}>C_\alpha(M_\ell^{s+1}-1)^{r+1}\] 
for some set of primes $\mathcal A\subseteq \{p_1,\dots, p_\ell\}$. This also shows that the code has distance at least $n-\# \mathcal A+1$. We now show that if we choose $\mathcal A$ to be the set of all primes of $\mathcal O_K$ lying above $p_{\lfloor c\ell \rfloor+1},\ldots,p_\ell$, the hypothesis of Proposition \ref{prop:distance} are satisfied.
By multiplying both sides of the inequality $1>k\sqrt[r+1]{C_\alpha}$ by $P_\ell/P_{\lfloor c\ell \rfloor} $ we obtain that:
\[P_\ell/P_{\lfloor c\ell \rfloor}=\prod\limits_{i=\lfloor c\ell \rfloor+1}^{\ell} p_i >\sqrt[r+1]{C_\alpha} \frac{kP_\ell}{P_{\lfloor c\ell \rfloor}}.\]
By raising both sides to the $(r+1)$-th power we get that, since $N(\mathfrak p_j^{(p_i)})=p_i$ for all $j$'s,
\begin{align*}
N\left(\prod\limits_{i=\lfloor c\ell \rfloor+1}^{\ell}\prod\limits^{r+1}_{j=1} \mathfrak p_j^{(p_i)} \right)=\prod\limits_{i=\lfloor c\ell \rfloor+1}^{\ell} p_i^{r+1}> & C_\alpha \left(\frac{kP_\ell}{P_{\lfloor c\ell \rfloor}}\right)^{r+1} \\ \geq & C_{\alpha} M_\ell^{(s+1)(r+1)} \\ 
> & C_\alpha(M_\ell^{s+1}-1)^{r+1},
\end{align*}
proving both that the code is good, that the distance tends to infinity as $\ell$ grows.
%\begin{align*}
%=&  \left(\frac{P_\ell}{P_{\lfloor c\ell \rfloor}}\right)^{r+1} > \left(\frac{kP_\ell}{P_{\lfloor c\ell \rfloor}}\right)^{r+1} \\
%\geq & 
%\end{align*}
%C_\alpha(M_\ell^{s+1}-1)^{r+1}

Next, we need to prove that the rate of $\mathcal C_\ell$ tends to a constant greater than zero, i.e.\ \eqref{good2}. Let $R_\ell\coloneqq \prod\limits_{i=1}^\ell\F_{p_i}^{r+1}$. Then
\begin{align*}
\liminf_{\ell\to +\infty}\frac{\log{\# \mathcal C_\ell}}{\log{\# R_\ell}} &=\liminf_{\ell\to +\infty}\frac{r(s+1)\log{\left\lfloor\sqrt[s+1]{k\cdot P_\ell/P_{\lfloor c\ell\rfloor}}\right\rfloor}}{(r+1)\log{P_\ell}}\\
&\ge\liminf_{\ell\to +\infty}\frac{r(s+1)\log\left(\sqrt[s+1]{k\cdot P_\ell/P_{\lfloor c\ell\rfloor}}-1\right)}{(r+1)\log{P_\ell}}\\
&\geq \liminf_{\ell\to +\infty}\frac{r\log\left(k\cdot P_\ell/P_{\lfloor c\ell\rfloor}\right)}{(r+1)\log{P_\ell}}.
\end{align*}

Now, using the properties of logarithms and the fact that $\log(P_\ell)\sim \ell\log \ell$ thanks to Lemma \ref{lemma:asymptotictotsplit} we get
\begin{align*}
\liminf_{\ell\to +\infty}\frac{r\log\left(k\cdot P_\ell/P_{\lfloor c\ell\rfloor}\right)}{(r+1)\log{P_\ell}} &= \liminf_{\ell\to +\infty} \frac{r(\log k+\ell \log \ell-c\ell  \log c\ell)}{(r+1)\ell  \log \ell}\\
&=\frac{r(1-c)}{r+1},
\end{align*}

satisfying \eqref{good1}.
\end{proof}

\begin{remark}
Notice that the distance grows linearly in $\ell$ (which is itself proportional to length and dimension), that is a desirable code property.
\end{remark}

\section{Realization of the construction}

By the Kronecker-Weber theorem, one can always construct a Galois extension $K/\Q$ of degree $r+1$ such that $\Gal(K/\Q)$ is cyclic of order $r+1$. This guarantees that the construction is always feasible and the number of totally split places is ``large'' (as their density will be roughly asymptotic to $1/(r+1)$).

The lemma that follows provides a constructive proof for the following curious (but expected) fact, for which we could not find reference in the literature: given a positive integer $\delta$ and $n$ rational primes $p_1,\ldots,p_n$ larger than $\delta$ it is always possible to construct explicitly a number field of degree $\delta$ where $p_1,\ldots,p_n$ are all totally split. This shows that if one desires to construct a locally recoverable code over a certain fixed product of finite fields, this is in theory possible.

\begin{lemma}
Let $\delta\in \Z_{>1}$ and let $p_1,\ldots,p_n$ be distinct rational primes all larger than $\delta$. Then it is possible to explicitly construct a monic, irreducible polynomial $f(x)\in \Z[x]$ of degree $\delta$ such that if $\alpha$ is a root of $f$ then all the $p_i$'s are totally split in the number field $\Q(\alpha)$.
\end{lemma}
\begin{proof}
For each $i\in \{1,\dots,n\}$, choose $\alpha_1^i,\ldots,\alpha_\delta^i\in \Z$ such that $\alpha_j^i\not\equiv \alpha_k^i\bmod p_i$ for every $j\neq k$ (this is possible because $p_i>\delta$). Next, choose a new prime $p_{n+1}$, different from $p_1,\ldots,p_n$, and for every $i\in \{1,\ldots,n+1\}$ let $q_i\coloneqq \prod_{j\neq i}p_j$. Notice that $q_1+\ldots+q_n$ is coprime with $q_{n+1}$, as if a prime $p$ divides $q_{n+1}$ then $p=p_i$ for some $i\in \{1,\ldots,n\}$ and hence $p$ divides $q_j$ for every $j\in \{1,\ldots,n\}\setminus \{i\}$; it follows that $p$ does not divide $q_1+\ldots+q_n$. Hence there exist $u_1,u_2\in \Z$ with $u_1u_2\ne 0$ such that $u_1(q_1+\ldots+q_n)+u_2q_{n+1}=1$. Notice that $u_1$ is coprime with $p_1\cdot\ldots\cdot p_n$ and $u_2$ is coprime with $p_{n+1}$. Now let $g(x)\in \Z[x]$ be a monic degree $\delta$ polynomial that is irreducible modulo $p_{n+1}$.

Consider then the polynomial
$$f(x)=u_1\sum_{i=1}^nq_i\prod_{j=1}^\delta(x-\alpha_j^i)+u_2q_{n+1}g(x)\in \Z[x].$$
By construction, $f(x)$ is monic. Moreover $f(x)$ is irreducible in $\Z[x]$ because $f(x)$ is irreducible modulo $p_{n+1}$ (because its reduction is $g(x) \mod p_{n+1}$, which is irreducible by construction), and hence the number field $K$ generated by a root $\alpha$ of $f$ has degree $\delta$. It remains to show that $p_1,\ldots,p_n$ are totally split in $K$. But this follows immediately from Dedekind criterion, that can be applied because none of the $p_i$'s divide the discriminant of $f$ since $f$ has no multiple roots modulo any $p_i$. It follows that the factorization pattern of $p_i$ in $\O_K$ coincides with that of $f$ modulo $p_i$; by construction this is a product of $\delta$ distinct linear terms.
\end{proof}

\section{Acknowledgements}

This work was supported in part by the National Science Foundation under Grant No 2127742.
Dorian Goldfeld is partially supported by Simons Foundation Grant Number 567168.

%Now by iterating B\'ezout's identity it is easy to find $u_1,\ldots,u_{n+1}\in \Z$ such that $\prod_{i=1}^{n+1} u_i\neq 0$ and $\sum_{i=1}^{n+1}u_iq_i=1$ (ADD SOME DETAILS). 
\bibliographystyle{plain}
\bibliography{bibliography}

\begin{thebibliography}{10}

\bibitem{barg2017locally}
Alexander Barg, Kathryn Haymaker, Everett~W Howe, Gretchen~L Matthews, and
  Anthony V{\'a}rilly-Alvarado.
\newblock Locally recoverable codes from algebraic curves and surfaces.
\newblock In {\em Algebraic Geometry for Coding Theory and Cryptography}, pages
  95--127. Springer, 2017.

\bibitem{bartoli2020locally}
Daniele Bartoli, Maria Montanucci, and Luciane Quoos.
\newblock Locally recoverable codes from automorphism group of function fields
  of genus $g \geq 1$.
\newblock {\em IEEE Transactions on Information Theory}, 66(11):6799--6808,
  2020.

\bibitem{magma}
Wieb Bosma, John Cannon, and Catherine Playoust.
\newblock The {M}agma algebra system. {I}. {T}he user language.
\newblock {\em J. Symbolic Comput.}, 24(3-4):235--265, 1997.
\newblock Computational algebra and number theory (London, 1993).

\bibitem{dukes2022optimal}
Austin Dukes, Andrea Ferraguti, and Giacomo Micheli.
\newblock Optimal selection for good polynomials of degree up to five.
\newblock {\em Designs, Codes and Cryptography}, 90(6):1427--1436, 2022.

\bibitem{freij2016locally}
Ragnar Freij-Hollanti, Thomas Westerb{\"a}ck, and Camilla Hollanti.
\newblock Locally repairable codes with availability and hierarchy: matroid
  theory via examples.
\newblock In {\em International Zurich Seminar on Communications-Proceedings},
  pages 45--49. ETH Zurich, 2016.

\bibitem{guruswami}
Venkatesan Guruswami.
\newblock Constructions of codes from number fields.
\newblock {\em IEEE Trans. Inform. Theory}, 49(3):594--603, 2003.

\bibitem{lenstra}
H.~W. Lenstra, Jr.
\newblock Codes from algebraic number fields.
\newblock In {\em Mathematics and computer science, {II} ({A}msterdam, 1986)},
  volume~4 of {\em CWI Monogr.}, pages 95--104. North-Holland, Amsterdam, 1986.

\bibitem{liu2018new}
Jian Liu, Sihem Mesnager, and Lusheng Chen.
\newblock New constructions of optimal locally recoverable codes via good
  polynomials.
\newblock {\em IEEE Transactions on Information Theory}, 64(2):889--899, 2018.

\bibitem{micheliIEEE}
Giacomo Micheli.
\newblock Construction of locally recoverable codes which are optimal.
\newblock {\em IEEE transactions on information theory}, 66(1):167--175, 2020.

\bibitem{silberstein2013optimal}
Natalia Silberstein, Ankit~Singh Rawat, O~Ozan Koyluoglu, and Sriram
  Vishwanath.
\newblock Optimal locally repairable codes via rank-metric codes.
\newblock In {\em 2013 IEEE International Symposium on Information Theory},
  pages 1819--1823. IEEE, 2013.

\bibitem{tamo2014family}
Itzhak Tamo and Alexander Barg.
\newblock A family of optimal locally recoverable codes.
\newblock {\em IEEE Transactions on Information Theory}, 60(8):4661--4676,
  2014.

\end{thebibliography}
\end{document}